\documentclass[pra,showpacs,graphics,twocolumn,floatfix,mathbbm,a4paper,nofootinbib]{revtex4-2}
\usepackage{amsthm}
\usepackage{amsmath}
\usepackage{latexsym}
\usepackage{amsfonts}
\usepackage{amssymb}
\usepackage{color}
\usepackage{bbm,dsfont}
\usepackage{graphicx}
\usepackage{subfigure}
\usepackage{mathbbol}
\usepackage{hyperref}
\usepackage{enumerate}
\usepackage{MnSymbol}

\newtheorem{proposition}{Proposition}

\newtheorem{lemma}{Lemma}
\newtheorem{corollary}{Corollary}

\theoremstyle{definition}

\newtheorem{example}{Example}
\newtheorem{definition}{Definition}

\newcommand{\complex}{\mathbb C} 
\renewcommand{\nat}{\mathbb N} 

\newcommand{\hi}{\mathcal{H}} 
\newcommand{\hik}{\mathcal{K}} 
\newcommand{\hv}{\mathcal{V}} 
\newcommand{\lh}{\mathcal{L(H)}} 
\newcommand{\lsh}{\mathcal{L}_s(\mathcal{H})} 
\newcommand{\lk}{\mathcal{L(K)}} 
\newcommand{\lv}{\mathcal{L(V)}} 
\newcommand{\sh}{\mathcal{S(H)}} 
\newcommand{\shik}{\mathcal{S(K)}} 
\newcommand{\shv}{\mathcal{S(V)}} 
\newcommand{\eh}{\mathcal{E(H)}} 
\newcommand{\ip}[2]{\left\langle\,#1\,|\,#2\,\right\rangle} 
\newcommand{\ket}[1]{|#1\rangle} 
\newcommand{\bra}[1]{\langle#1|} 
\newcommand{\kb}[2]{|#1\rangle\langle#2|} 
\newcommand{\tr}[1]{\mathrm{tr}\left[#1\right]} 

\newcommand{\obs}{\mathcal{O}}
\newcommand{\A}{\mathsf{A}}
\newcommand{\B}{\mathsf{B}}

\newcommand{\T}{\mathsf{T}}

\newcommand{\C}{\mathcal{C}}
\newcommand{\D}{\mathcal{D}}
\newcommand{\E}{\mathcal{E}}
\newcommand{\G}{\mathcal{G}}
\newcommand{\Q}{\mathcal{Q}}
\newcommand{\R}{\mathcal{R}}
\newcommand{\I}{\mathcal{I}}
\newcommand{\J}{\mathcal{J}}
\newcommand{\K}{\mathcal{K}}
\newcommand{\ins}{\mathrm{Ins}}

\newcommand{\simu}[1]{\mathfrak{sim}(#1)} 

\begin{document}

\title[]{Post-processing of quantum instruments}

\author{Leevi Lepp\"{a}j\"{a}rvi}
\email{leille@utu.fi}
\affiliation{QTF Centre of Excellence, Turku Centre for Quantum Physics, Department of Physics and Astronomy, University of Turku, Turku 20014, Finland}

\author{Michal Sedl\'ak}
\email{michal.sedlak@savba.sk}
\affiliation{RCQI, Institute of Physics, Slovak Academy of Sciences, D\'ubravsk\'a cesta 9, 84511 Bratislava, Slovakia}
\affiliation{Centre of Excellence IT4Innovations, Faculty of Information Technology, Brno University of Technology, Bo\v zet\v echova 2/1, 612 00 Brno, Czech Republic}



\begin{abstract}
Studying sequential measurements is of the utmost importance to both the foundational aspects of quantum theory and the practical implementations of quantum technologies, with both of these applications being abstractly described by the concatenation of quantum instruments into a sequence of certain length. In general, the choice of instrument at any given step in the sequence can be conditionally chosen based on the classical results of all preceding instruments. For two instruments in a sequence we consider the conditional second instrument as an effective way of post-processing the first instrument into a new one. This is similar to how a measurement described by a positive operator-valued measure (POVM) can be post-processed into another by way of classical randomization of its outcomes using a stochastic matrix. In this work we study the post-processing relation of instruments and the partial order it induces on their equivalence classes. We characterize the greatest and the least element of this order, give examples of post-processings between different types of instruments and draw connections between post-processings of some of these instruments and their induced POVMs.
\end{abstract}

\maketitle

\section{Introduction}
Measurements are the most essential part of every physical theory as they are, by definition, the only way to extract information from the objects that the theory describes. Detailed understanding and characterization of measurements may reveal crucial features of the theory. This is especially true in quantum theory where the discoveries of phenomena such as inevitable information/disturbance trade-off or intrinsic randomness of measurement outcomes has made our (classical) intuition imprecise and it renders quantum theory as puzzling to most physicist even after hundred years.

With the bloom of the field of quantum computation and information processing measurements are viewed not only as a mere tool for making predictions and verification of the theory, but they are also used as the key element assuring the implementation of the given task. For example, in quantum teleportation \cite{teleportation} they enable transmission of quantum information, but they can act as a step in quantum computation as well \cite{mbcomp}, or as the decoder of classical information in its transmission via noisy quantum channels \cite{helstrom, HolevoBook}.

Most often quantum measurements are considered in two ways. If we are interested only in the classical outcome produced by the quantum measurement then it is conveniently described by a positive operator-valued measure (POVM) \cite{BuschLahtiMittelstaedtBook}. However, if the measured quantum system will be part of any further interaction then we shall use quantum instrument \cite{BuschLahtiMittelstaedtBook} to describe also the state change induced by the measurement. Especially if we want consider evolutions, where measurements of the same system happen at different times, description of sequential measurements by quantum instruments becomes really practical.

Sequences of quantum measurements are at the same time very fundamental and application fruitful object to study. They can be used for example for quantum state tomography \cite{CaHeTo12, SNFilippov1} and estimation, property testing \cite{HarrowPropTest}, computation \cite{acinComp1}, quantum sequential decoding \cite{seqdecWilde}, joint measurability \cite{HeMiZi16} or sequential state discrimination \cite{bergou1}.

Let us illustrate the practical relevance of sequences of measurements with two of the previously mentioned applications. First, suppose we would be repeating the same instrument in the sequence. If the instrument corresponds to a least disturbing realization of a non-degenerate projective measurement (its L\"uder's instrument) then such a repetition would not be useful at all, because the sequence of outcomes would be just a repetition of the first measurement outcome. However, in practical implementations measurements are not ideal and will always be noisy. Recently, in \cite{HHKsaturation} the authors discovered that repetition of such noisy measurement can lead to an effective measurement with suppressed noise level. Secondly, the authors of \cite{HeMi15uni} showed that any jointly measurable pair of observables can be jointly measured by a L\"uder's instrument of the first observable and a subsequent measurement of its output by a suitably chosen measurement. Hence, we can also say that joint measurability can be operationally realized via sequential measurements.

In the current work we study a (concatenated) sequence of two quantum instruments, i.e. two quantum measurements having both classical and quantum output such that the classical outcome of the first instrument determines the second quantum instrument that will further process the quantum output of the first instrument. Thus, the concatenation of the first quantum instrument with the conditionally selected second instrument leads to a new compound quantum instrument. This can be viewed as a post-processing of the first instrument into the resulting compound quantum instrument. If one fixes the first quantum instrument and considers all possible post-proceesing quantum instruments one can understand how an action of the first already performed instrument can be altered or modified. In particular, one may also try to answer when it can be completely reversed.

The post-processing relation can be used to define a partial order between (equivalence classes of) quantum instruments: if an (equivalence class of an) instrument can be post-processed from some other (equivalence class of an) instrument, then the former must be greater than the latter. From the resource theory perspective this can be interpreted as statement that the resulting instrument is less useful or less of a resource than the instrument that is being post-processed. Special case of single outcome quantum instruments corresponds to post-processing of quantum channels, which was previously studied in \cite{HeMi13}.

The post-processing that we define is analogous to the case of post-processing of POVMs, where deterministic post-processing of the classical outcomes of a POVM lead to a new POVM. In the case of quantum instruments, instead of just post-processing the classical outcome we should also consider the post-processing of the post-measurement state and in general these two aspects do not need to be independent. Thus, we are not restricted to changing the classical outcome by a classical post-processing matrix and modifying the post-measurement state by a quantum channel. In general, we can choose a different post-processing quantum instrument for every outcome of the original instrument.

Study of the partial order induced by the post-processing relation is one of the main goals of this manuscript. For general quantum instruments we characterize the least and the greatest element of the partial order and characterize their equivalence classes. We give examples of post-processings of various types of instruments and draw connections between post-processing of quantum instruments and post-processing of their induced POVMs. Finally, we apply the post-processing as a part of the simulation scheme for instruments, where in addition to post-processing a given set of instruments we are also allowed to classically mix them. We note that quantum instruments that map to one-dimensional output space are isomorphic to POVMs. Thus, we must recover POVM simulability \cite{GuBaCuAc17, OsGuWiAc17, FiHeLe18, OsMaPu19} as a special case of our results. In this sense our work can be also seen as a generalization of POVM simulability.

The rest of the manuscript is organized as follows. In Sec. \ref{sec:prel} we introduce the relevant concepts and notation. In Sec. \ref{sec:POVMs} we first consider the post-processing of POVMs  before generalizing this concept to instruments in Sec. \ref{sec:instruments}. In this section we also analyze the structure given by the post-processing relation. In Sec. \ref{sec:induced-POVM} we show connections between the post-processings of instruments and their induced POVMs.  Finally, in Sec. \ref{sec:simulation} we focus on the concept of simulation of instruments using the post-processing relation defined in Section \ref{sec:instruments}. Our findings are summarized in Sec. \ref{sec:summary}.


\section{Preliminaries and notation} \label{sec:prel}
Let $\hi$ be a finite-dimensional complex Hilbert space. We denote by $\lh$ the set of bounded operators on $\hi$ and by $\lsh$ the set of selfadjoint operators in $\lh$. States of a quantum system are described by positive semi-definite operators on $\hi$ with unit trace. We denote the set of quantum states on $\hi$ by $\sh$ so that
\begin{equation*}
\sh= \{ \varrho \in \lsh \, | \, \varrho \geq 0, \ \tr{\varrho} =1 \},
\end{equation*}
where $0$ is the zero operator.

The set of effects on $\hi$, denoted by $\eh$, consists of selfadjoint operators on $\hi$ bounded by $0$ and $I$, where $I$ (or $I_\hi$ if we want to be more specific) is the identity operator on $\hi$, i.e.,
\begin{equation*}
\eh = \{E \in \lsh \, | \, 0 \leq E \leq I \}.
\end{equation*}
Observable with a finite number of outcomes is described by a \textit{positive operator-valued measure} (POVM), i.e. by mapping $\A: x \mapsto \A(x)$ from a finite outcome set $\Omega$ to $\eh$ such that $\sum_{x \in \Omega} \A(x) = I$. The set of observables on $\hi$ with outcome set $\Omega$ is denoted by $\obs(\Omega, \hi)$.

Let $\hi$ and $\hik$ be Hilbert spaces. Transformations of states on $\hi$ to states on $\hik$ are described by \textit{quantum channels}, i.e., completely positive trace-preserving maps from $\lh$ to $\lk$. Probabilistic transformations are described by \textit{quantum operations}, i.e., completely positive trace-nonincreasing maps.

Quantum channels and operations have a well-known represention in an operator-sum form: a linear map $\mathcal{N}: \lh \to \lk$ is a quantum operation if and only if there exists bounded operators $K_i: \hi \to \hik$ for all $i=1,2, \ldots $ such that $\mathcal{N}(\varrho) = \sum_i K_i \varrho K_i^*$ for all $\varrho \in \lh$ and
$\sum_i K_i^* K_i \leq I$. The operators $K_i$ are called \textit{Kraus operators} of $\mathcal{N}$ and in the finite-dimensional case it is possible to choose $\dim(\hi) \dim(\hik)$ or fewer Kraus operators. The minimal number of Kraus operators for a given operation is called the \textit{Kraus rank} of the operation.

\textit{Quantum instrument} describes a device that takes a quantum input state and gives a quantum output state conditioned on a classical measurement outcome. Formally, a quantum instrument $\I$ is a mapping $\I: x \mapsto \I_x$ from a finite outcome set $\Omega$ to the set of operations such that $\sum_{x \in \Omega} \I_x $ is a quantum channel. Given an input $\varrho \in \sh$, the (unnormalized) conditional output state is then described by $\I_x(\varrho)$ when we get the outcome $x$ in the measurement of the \emph{induced POVM} $\A^\I \in \obs(\Omega, \hi)$ described by the probabilities $\tr{\A^\I(x)\varrho} = \tr{\I_x(\varrho)}$. In general it is clear that different instruments can have the same induced POVM but that for any given instrument the induced POVM is unique. The set of instruments from $\lh$ to $\lk$ with outcome set $\Omega$ is denoted by $\ins(\Omega, \hi, \hik)$. In the case when the input and the output spaces are the same, $\hik = \hi$, we denote the set simply $\ins(\Omega, \hi)$.

\begin{example}
We introduce the \emph{identity instrument} $id: \lh \to \lh$ as a 1-outcome instrument that leaves the state unchanged, i.e., $id(\varrho) = \varrho$ for all $\varrho \in \sh$. Because this instrument has only one outcome, it is in fact a channel, and thus we also refer to it just as the identity channel. One immediately sees that the identity channel is reversible and in fact one can consider it as a special case in the class of \emph{unitary channels} $\mathcal{U}: \lh \to \lh$, which are reversible channels defined by some unitary operator $U$ on $\hi$ as $\mathcal{U}(\varrho) = U \varrho U^*$ for all $\varrho \in \sh$.
\end{example}

\begin{example}
Another class of 1-outcome instruments (i.e. channels) are the \emph{trash-and-prepare channels} $\mathcal{T}: \lh \to \lk$ that are defined as $\mathcal{T} (\varrho) = \tr{\varrho}\xi$ for all $\varrho \in \lh$ for some fixed state $\xi$. Trash-and-prepare channels are also sometimes called complete state-space contractions since they just ignore the input state and prepare a new fixed state, i.e., the whole state space is contracted into a single point.

One can also consider trash-and-prepare instruments with more outcomes simply by trashing the input state, rolling a dice and preparing a new state based on the outcome of the dice roll. Thus, we can have a trash-and-prepare instrument $\mathcal{T}$ with outcome set $\Omega$ defined as $\mathcal{T}_x(\varrho)= \tr{\varrho} p_x \xi_x$ for all $x \in \Omega$, where $(p_x)_x$ is some probability distribution over $\Omega$ and $\{\xi_x\}_x$ is a set of states. Clearly the channel corresponding to this instrument is a trash-and-prepare channel that outputs the mixed state $\xi = \sum_{x \in \Omega} p_x \xi_x$. Thus, it can be seen as a convex mixture of other trash-and-prepare channels.
\end{example}

\begin{example}
Instead of preparing a new state just by trashing the input state, one can also perform a (demolishing) measurement on the input and then prepare a new state according to the measurement outcome. Thus, if $\A \in \obs(\Omega, \hi)$ is a POVM, we can define a \emph{measure-and-prepare instrument} $\mathcal{P}^\A \in \ins(\Omega, \hi, \hik)$ for some set of states $\{\xi_x\}_{x \in \Omega}$ as $\mathcal{P}^\A_x(\varrho) = \tr{\A(x)} \xi_x$ for all $\varrho \in \sh$. Note that by changing the set $\{\xi_x\}_{x \in \Omega}$, one can use the same POVM $\A$ to define countless measure-and-prepare instruments. One also sees that the trash-and-prepare instrument is a special case of a measure-and-prepare instrument where one just fixes the POVM to be trivial, i.e., $\A(x) = p_x I$ for all $x \in \Omega$ for some probability distribution $(p_x)_x$ over $\Omega$.
\end{example}

\section{Post-processing of POVMs} \label{sec:POVMs}
Before generalizing the concept of post-processing to instruments, we recall some important results for POVMs. After obtaining the outcome statistics of a measurement of an observable, one may want to process the obtained information. One can, for instance, see if it is possible to reveal some other property of the system by manipulating the data and obtain the outcome statistics of some other observable. This is what is usually called the post-processing of observables.

\subsection{The post-processing partial order}

We can formalize the previous paragraph with the following definition.

\begin{definition}
Let $\A \in \obs(\Omega_\A, \hi)$ and $\B \in \obs(\Omega_\B, \hi)$ be observables. If there exists a stochastic matrix $\nu = \left( \nu_{xy} \right)_{x \in \Omega_\A, y \in \Omega_\B}$,  i.e., $\nu_{xy} \geq 0$ for all $x \in \Omega_\A$, $y \in \Omega_\B$, and $\sum_{y \in \Omega_\B} \nu_{xy} =1 $ for all $x \in \Omega_\A$, such that
\begin{equation*}
\B(y) = \sum_{x \in \Omega_\A} \nu_{xy} \A(x)
\end{equation*}
for all $y \in \Omega_\B$, we say that $\B$ is a \textit{post-processing} of $\A$ and denote it $\A \to \B$. Furthermore, we say that observables $\A$ and $\B$ are \textit{post-processing equivalent}, denoted by $\A \leftrightarrow \B$, if $\A \to \B$ and $\B \to \A$.
\end{definition}

\begin{example}
A special kind of post-processing, called \emph{relabeling}, is one where all the elements of the stochastic post-processing matrix are either $0$ or $1$. Following  \cite{HaHeMi18}, this can be formalized by the existence of a function $f: \Omega_\A \to \Omega_\B$ such that $\nu_{xy} = \delta_{f(x),y}$, where $\delta_{x,x'}$ is the Kronceker delta, so that
\begin{equation*}
\B(y) = \sum_{x \in f^{-1}(y)} \A(x)
\end{equation*}
for all $y \in \Omega_\B$. In this case, we say that $\B$ is a relabeling of $\A$ and that $\A$ is a \emph{refinement} of $\B$.
\end{example}

Post-processing captures the idea that the outcome statistics of $\B$ can be deterministically obtained from the statistics of $\A$ by some classical process represented by the stochastic matrix. It is easy to see that post-processing induces a preorder on the set of all observables on $\hi$, and by extending it to the equivalence classes of post-processing equivalent observables it becomes a partial order. A natural thing to consider is whether there exist a least or greatest element with respect to this order.

It is easy to see that the trivial observables, i.e., observables $\T^p \in \obs(\Omega_{\T^p}, \hi)$ of the form $\T^p(x) = p_x I$ for all $x \in \Omega_{\T^p}$ for some probability distribution $(p_x)_x$ on $\Omega_{\T^p}$, can be post-processed from any other observable $\A \in \obs(\Omega_\A, \hi)$ by using the post-processing matrix $\nu$ with $\nu_{xy} = p_x$ for all $x \in \Omega_{\T^p}$  and $y \in \Omega_\A$. Thus, $\A \to \T^p$ for any observable $\A$ and any probability distribution $p$. Furthermore, if $\B$ is an observable such that $\T^p \to \B$ with a post-processing $\mu$, then also $\B$ is a trivial observable, $\B = \T^q$, where $q_z = \sum_{x \in \Omega_{\T^p}} \mu_{xz} p_x$ for all $z \in \Omega_\B$. Thus, the equivalence class of trivial observables is the least element with respect to the partial order.

How about the greatest element? Turns out that there is no greatest element \cite{HeMi13, MaMu90}. Instead, we get a class of maximal elements that we call post-processing clean observables.

\begin{definition}
An observable $\A$ is \emph{post-processing clean} if for any observable $\B$ such that $\B \to \A$ we also have $\A \to \B$.
\end{definition}

The post-processing clean observables were characterized in \cite{MaMu90, BuDAKePeWe05}: an observable is post-processing clean if and only if it is rank-1, i.e., each of its effects is a rank-1 operator. For a POVM $\A \in \obs(\Omega_\A, \hi)$, being rank-1 is equivalent to being \emph{indecomposable} \cite{KiNuIm10}, i.e., if any of its non-zero effect $\A(x)$ is decomposed as a sum of some two effects on $\hi$ so that $\A(x) = E_x +F_x$ for some $E_x,F_x \in \eh$, then there exist positive numbers $e_x,f_x >0$ such that $\A(x) = e_x E_x = f_x F_x$. It holds that any observable can be post-processed from a post-processing clean observable. Thus, rank-1 POVMs are in fact the maximal elements with respect to the post-processing partial order as everything else can be post-processed from them.

\subsection{Minimally sufficient POVMs}
Post-processing can thus be seen as a way to construct new observables out of existing ones by a classical process. Another way to look at post-processing is to say that if we have $\A \to \B$ for two observables $\A$ and $\B$, then $\A$ must be more informative as $\B$ can be deduced from $\A$. But as was pointed out earlier, when we talk about the partial order induced by the post-processing relation, we are actually comparing equivalence classes of observables. Then especially post-processing equivalent observables would be just as informative. However, even though they can be seen as having the same information, the following notion introduced in \cite{Kuramochi15} captures the idea that even in the same equivalence class there are observables with minimum informational redundancy:

\begin{definition}
An observable $\A$ is \emph{minimally sufficient} if, whenever $\A \leftrightarrow \B$ with some observable $\B$, then $\B$ is a refinement of $\A$.
\end{definition}

It was shown in \cite{Kuramochi15} that a (discrete) POVM $\A \in \obs(\Omega_\A, \hi)$ is minimally sufficient if and only if it is \emph{non-vanishing}, i.e., $\A(x) \neq 0$ for all $x \in \Omega_\A$, and it is \emph{pairwise linearly independent}, i.e., $\A(x) \neq c \A(y)$ for any $c >0$ for all $x \neq y$, $x,y \in \Omega_\A$. Furhtermore, for any POVM $\A$, there exists a minimally sufficient POVM $\tilde{\A}$ such that $\A \leftrightarrow \tilde{\A}$, and $\tilde{\A}$ is unique up to a bijective relabeling of its outcomes.

The minimally sufficient representative of the equivalence class of a POVM $\A$ can be constructed as follows: define an equivalence relation $\sim$ in $\Omega_\A$ so that $x \sim y$ if and only if there exists $c>0$ such that $\A(x) = c \A(y)$. We denote the set of equivalence classes $\Omega_\A /\hspace*{-0.1cm}\sim $ by $\tilde{\Omega}_\A$ and define a minimally sufficient POVM $\tilde{\A} \in \obs(\tilde{\Omega}_\A, \hi)$ that is post-processing equivalent to $\A$ by
\begin{equation*}
\tilde{\A}([y]) = \sum_{x \in [y]} \A(x), \quad \quad [y] \in \tilde{\Omega}_\A.
\end{equation*}

The uniqueness of the pairwise linearly independent minimally sufficient representative can be used to characterize the whole post-processing equivalence class: two POVMs $\A $ and $\B $ are post-processing equivalent if and only if the pairwise linearly independent POVMs $\tilde{\A}$ and $\tilde{\B}$ are bijective relabelings of each other. This shows that the effects of two post-processing equivalent POVMs must be proportional to each other.

\begin{proposition}\label{prop:pp-equivalent-POVMs}
Let $\A \in \obs(\Omega_\A, \hi)$ and $\B \in \obs(\Omega_\B, \hi)$ be two post-processing equivalent non-vanishing POVMs. Then for all $x \in \Omega_\A$ there exists $y_x \in \Omega_\B$ and $c_{xy_x}>0$ such that $\A(x) = c_{xy_x} \B(y_x)$. Furthermore, there exist post-processings $\nu$ for $\B \to \A$ and $\mu$ for $\A \to \B$ such that $\nu_{yx}, \mu_{xy} \neq 0$ only if $\B(y)$ is proportional to $\A(x)$.
\end{proposition}
\begin{proof}
We define the pairwise linearly independent POVMs $\tilde{\A} \in \obs(\tilde{\Omega}_\A, \hi)$ and $\tilde{\B} \in \obs(\tilde{\Omega}_\B, \hi)$ as above, so that
\begin{equation*}
\tilde{\A}([x]) = \sum_{x' \in [x]} \A(x'), \quad \tilde{\B}([y]) = \sum_{y' \in [y]} \B(y')
\end{equation*}
for all $[x] \in \tilde{\Omega}_\A$ and $[y] \in \tilde{\Omega}_\B$. Thus, for all $x' \in [x]$, we have that $\A(x') = c_{x'} \tilde{\A}([x])$ for some $c_{x'}\in (0,1]$ such that $\sum_{x' \in [x]} c_{x'} = 1$. Similarly, for all $y' \in [y]$, we have that $\B(y') = d_{y'} \tilde{\B}([y])$ for some $d_{y'}\in (0,1]$ such that $\sum_{y' \in [y]} d_{y'} = 1$.

As was mentioned earlier, the minimally sufficient representative is essentially unique in each equivalence class, so that since $\tilde{\A} \leftrightarrow \A \leftrightarrow \B \leftrightarrow \tilde{\B}$ and since $\tilde{\A}$ and $\tilde{\B}$ are both minimally sufficient, there exists a bijective map $f: \tilde{\Omega}_\A \to \tilde{\Omega}_\B$ such that $\tilde{\A}([x]) = \tilde{\B}(f([x]))$ for all $[x] \in \tilde{\Omega}_\A$. Thus, for each $x \in \Omega_\A$ and $y_x \in f([x]) \subset \tilde{\Omega}_\B$ we have that
\begin{equation*}
\A(x) = c_x \tilde{\A}([x]) = c_x \tilde{\B}(f([x])) = c_x \tilde{\B}([y_x]) = c_{xy_x} \B(y_x),
\end{equation*}
where we have denoted $c_{xy_x} = c_x/d_{y_x} >0$.

For the second part of the claim, let us define $\nu_{yx} = c_x \delta_{f([x]), [y]}$ for all $x \in \Omega_\A$ and $y \in \Omega_\B$. Clearly for all $y \in \Omega_\B$ we have that
\begin{align*}
\sum_{x \in \Omega_\A} \nu_{yx} = \sum_{x \in f^{-1}([y])} c_x = 1
\end{align*}
since $f^{-1}([y]) = [x']$ for some $[x'] \in \tilde{\Omega}_\A$ for all $y \in \Omega_\B$. Furthermore,
\begin{align*}
\sum_{y \in \Omega_\B} \nu_{yx} \B(y) &= \sum_{y \in f([x])} c_x \B(y) = c_x \tilde{\B}(f([x])) = \A(x)
\end{align*}
for all $x \in \Omega_\A$. Thus, $\nu$ is a post-processing for $\B \to \A$ that has $\nu_{yx} \neq 0$ only if $\B(y)$ is proportional to $\A(x)$. The post-processing $\mu$ for $\A \to \B$ can be defined analogously.
\end{proof}

We note that not all POVMs whose effects are proportional to each other are post-processing equivalent. For example, let us define two 4-outcome qubit POVMs $\A$ and $\B$ as
\begin{align*}
\A(1) &= \frac{1}{2} \kb{\varphi_1}{\varphi_1}, \quad \A(2) = \frac{1}{2} \kb{\varphi_2}{\varphi_2}, \\
\A(3) &= \frac{1}{2} \kb{\psi_1}{\psi_1}, \quad \A(4) = \frac{1}{2} \kb{\psi_2}{\psi_2}, \\
\B(1)& = \frac{1}{3} \kb{\varphi_1}{\varphi_1}, \quad \B(2) = \frac{1}{3} \kb{\varphi_2}{\varphi_2}, \\
\B(3) &= \frac{2}{3} \kb{\psi_1}{\psi_1}, \quad \B(4) = \frac{2}{3} \kb{\psi_2}{\psi_2},
\end{align*}
where $\{\varphi_1, \varphi_2\}$ and $\{\psi_1, \psi_2\}$ are two orthonormal bases in $\complex^2$. We see that although the effects of $\A$ and $\B$ are proportional, they cannot be post-processed from each other since for instance $\A(1) = 3/2 \ \B(1)$ and $\B(3) = 4/3 \ \A(3)$, where obviously the post-processing elements $3/2$ and $4/3$ would be larger than one. Instead, as required by the bijective relabeling of two minimally sufficient representatives, two POVMs whose effects are proportional to each other are post-processing equivalent if their pairwise linearly dependent effects sum up to the same effect for both observables. For the above example this would mean that the coefficients of the rank-1 projectors would have to be the same for $\A$ and $\B$.

\section{Post-processing of instruments} \label{sec:instruments}
We can now define post-processing of instruments analogously to the post-processing of POVMs but unlike in the case of POVMs we are not only processing classical information but we must also process the output state of the instrument.

\begin{definition}\label{def:pp-instruments}
Let $\I \in \ins(\Omega, \hi, \hik)$ and $\J\in \ins(\Lambda, \hi, \hv)$ be quantum instruments. If there exists a set of instruments $\{\R^{(x)}\}_{x\in \Omega} \subset \ins(\Lambda, \hik, \hv) $ such that
\begin{equation}\label{eq:simulation}
\J_y(\varrho) = \sum_{x \in \Omega} \R^{(x)}_y \left( \I_x(\varrho) \right)
\end{equation}
for all $\varrho \in \lh$ and $y \in \Lambda$, then we denote $\I \to \J$ and say that $\J$ is a \emph{post-processing} of $\I$. Furthermore, we say that $\I$ and $\J$ are \emph{post-processing equivalent}, denoted by $\I \leftrightarrow \J$, if $\I \to \J$ and $\J \to \I$.
\end{definition}

\begin{figure}
\centering
\includegraphics[scale=0.32]{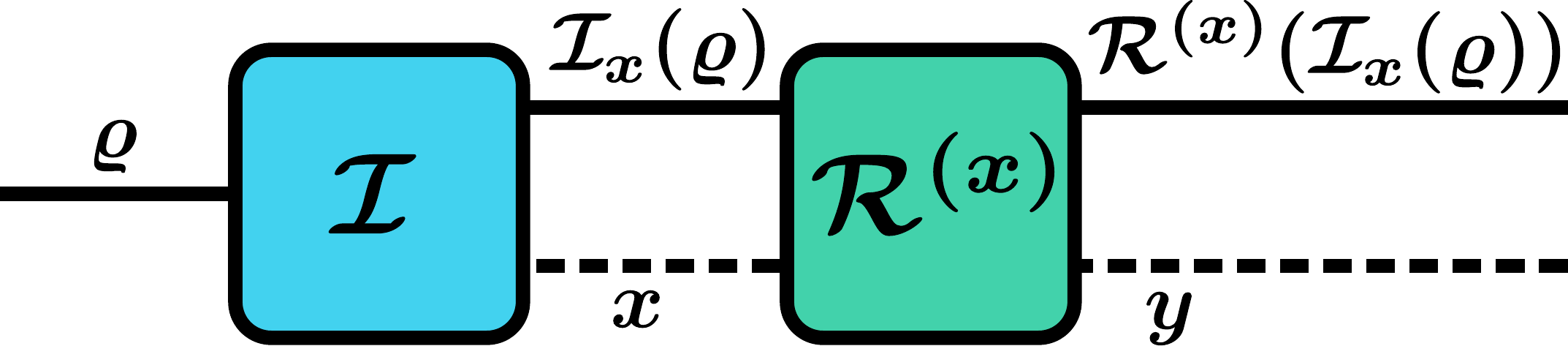} 
\caption{\label{fig:1} The post-processing of instrument $\I$ by the instruments $\R^{(x)}$. Instrument $\I$ takes quantum state $\varrho$ as an input and after measuring it and obtaining a measurement outcome $x$ the instrument $\R^{(x)}$ is chosen. The conditional output state of $\I$ (described by $\I_x(\varrho)$) serves as an input to $\R^{(x)}$, which measures it and produces an outcome $y$ after which the output is in the (unnormalized) state $\R^{(x)}_y(\I_x(\varrho))$. }
\end{figure}

The post-processing Eq. \eqref{eq:simulation} is depicted in Fig. \ref{fig:1}. As an illustrating example of post-processing of instruments, we consider what type of type of instruments can be post-processed from measure-and-prepare instruments.

\begin{example}\label{ex:m-a-p}
Let $\I \in \ins(\Omega, \hi, \hik)$ be a measure-and-prepare instrument with an induced POVM $\A$, i.e., it is of the form
\begin{align*}
\I_x(\varrho) &= \tr{\A(x) \varrho} \sigma_x
\end{align*}
for all $\varrho \in \sh$ and $x \in \Omega$ for some set of states $\{ \sigma_x\}_{x \in \Omega} \subset \shik$. Let $\J \in \ins(\Lambda, \hi, \hv)$ be an instrument such that $\I \to \J$ so that there exist instruments $\R^{(x)} \in \ins(\Lambda, \hik, \hv)$ such that $\J_y(\varrho)= \sum_x \R^{(x)}_y(\I_x(\varrho))$ for all $\varrho \in \sh$ and $y \in \Lambda$. By expanding the previous expression, we see that
\begin{align*}
 \J_y(\varrho) = \sum_{x \in \Omega} \tr{\A(x) \varrho} \R^{(x)}_y(\sigma_x)
\end{align*}
for all $y \in \Lambda$.
By denoting $\nu_{xy} = \tr{\R^{(x)}_y(\sigma_x)} \in [0,1]$ for all $x \in \Omega$ and $y \in \Lambda$, we can define $\xi_{xy} = \R^{(x)}_y(\sigma_x)/\nu_{xy} \in \shv$ when $\nu_{xy}\neq 0$ and $\xi_{xy} = \xi$ for some fixed $\xi \in \shv$ when $\nu_{xy} =0$. We see that $\sum_y \nu_{xy} =1$ for all $x \in \Omega$ so that $\nu$ is actually a valid post-processing. Thus, we have that
\begin{align*}
\J_y(\varrho) = \sum_{x \in \Omega} \tr{\nu_{xy} \A(x) \varrho} \xi_{xy}
\end{align*}
for all $y \in \Lambda$ and $\varrho \in \sh$. We note that the induced POVM $\A^\J$ of $\J$ is then a post-processing of $\A$ since by taking the trace on the last expression we see that $\A^\J(y) = \sum_{x \in \Omega} \nu_{xy} \A(x)$ for all $y \in \Lambda$ so that $\A = \A^\I \to \A^\J$.
\end{example}

We note that in the case of channels, Definition \ref{def:pp-instruments} reduces to the preorder given in \cite{HeMi13}: If $\C: \lh \to \lk$ and $\D: \lh \to \lv$ are two quantum channels such that there exists a channel $\E: \lk \to \lv$ such that $\D = \E \circ \C$, where $\circ$ denotes the composition of maps, then $\C \to \D$. Just as in the case of channels, the post-processing defined above gives a preorder in the set of instruments so that it can be used to define a partial order for the equivalence classes of instruments.

\subsection{Characterizing the greatest and the least element}

As in the case of POVMs, one of the first questions is to determine the greatest and the least element (if they exist) with respect to the post-processing partial order. Let us start with the greatest element, and as before let us first consider the maximal elements, namely, the post-processing clean instruments.

\begin{definition}
An instrument $\I$ is \emph{post-processing clean} if for any instrument $\J$ such that $\J \to \I$ we also have $\I \to \J$.
\end{definition}

Unlike in the case of POVMs, there is only one maximal element in the equivalence classes of instruments so that it must also be the greatest element. As expected, this is similar to the case of channels that was studied in \cite{HeMi13} and the greatest element is the same.

\begin{proposition}
An instrument is post-processing clean if and only if it is post-processing equivalent with the identity channel $id$.
\end{proposition}
\begin{proof}
First of all, it is clear that one can post-process every instrument with an input Hilbert space $\hi$ from the identity channel $id: \lh \to \lh$. Namely, if $\I \in \ins(\Omega, \hi, \hik)$ is any instrument, then $\I_x(\varrho) = \I_x(id(\varrho))$ for all $x \in \Omega$. By noting that the identity channel can be viewed as an instrument with only a single outcome, we see that the previous equation is of the same form as Eq. \eqref{eq:simulation}. Thus, $id \to \I$ for any instrument $\I$. Similarly, if $\J$ is an instrument that is post-processing equivalent with $id$ so that in particular $\J \to id$, then by the transitivity of the post-processing it follows that also $\J \to \I$ for any instrument $\I$. Moreover, if $\I$ is an instrument such that $\I \to \J$, then $\I$ must actually be post-processing equivalent with $\J$ (and with $id$) so that by definition $\J$ is post-processing clean.

Let then an instrument $\R \in \ins(\Omega, \hi, \hik)$ be post-processing clean. By the previous consideration, we have that $id \to \R$, so that because $\R$ is post-processing clean it follows that also $\R \to id$. Thus, any post-processing clean instrument is post-processing equivalent with $id$.
\end{proof}

From the proof of the previous Proposition we see that every instrument can be post-processed from the identity channel. We can conclude the following:

\begin{corollary}\label{cor:identity}
Every instrument can be post-processed from any instrument that is post-processing equivalent with the identity channel.
\end{corollary}

Furthermore, we can characterize the equivalence class of $id$.

\begin{proposition}\label{prop:id-equivalence-class}
An instrument $\I \in \ins(\Omega, \hi, \hik)$ is equivalent with the identity channel on $\hi$ if and only if for all $x \in \Omega$ and $\varrho \in \sh$ we have that
\begin{equation}\label{eq:isometries}
\I_x(\varrho) = \sum_{i=1}^{n_x} p_{xi} V_{xi} \varrho V_{xi}^*
\end{equation}
 for some probability distribution $(p_{xi})_{x,i}$ over $x\in \Omega$, $i\in \{1,\ldots,n_x\}$ and some isometries $V_{xi}: \hi \to \hik$ such that $V_{xj}^* V_{xi}=0$ for all $i\neq j$ for all $x\in \Omega$.
\end{proposition}

\begin{proof}
First, let $\I$ be of the form of Eq. \eqref{eq:isometries} for some probability distribution $p_{xi}$ and isometries $V_{xi}$. Since every instrument can be post-processed from the identity channel, we have that $id \to \I$. We define channels (1-outcome instruments) $\mathcal{R}^{(x)} \in \ins(\{0\}, \hik, \hi)$ by $\mathcal{R}^{(x)}(\varrho) = \sum_{i=1}^{n_x}  V_{xi}^* \varrho V_{xi} + \tr{\Pi_x\varrho}\ket{\psi}\bra{\psi}$ for all $x \in \Omega$ and $\varrho \in \shik$, where $\ket{\psi}$ is an arbitrary unit vector in $\hi$ and $\Pi_x=I_{\hik}-\sum_{i=1}^{n_x}  V_{xi} V_{xi}^* $. We remind that $\Pi_{xi}=V_{xi} V_{xi}^*$ are orthogonal projectors since $V_{xi}$ are isometries and $V_{xj}^* V_{xi}=0$ for $i\neq j$. If we denote the orthonormal vectors spanning the subspace on which $\Pi_x$ projects by $\ket{e_{xk}}$, for $k=1,\ldots,\tr{\Pi_x}$, and furthermore if we set $K_{xk}=\ket{\psi}\bra{e_{xk}}$, then $\tr{\Pi_x\varrho}\ket{\psi}\bra{\psi}=\sum_k K_{xk} \varrho K_{xk}^*$ and also $\{V_{xi}^*\}_{i=1}^{n_x} \cup \{K_{xk}\}_{k=1}^{\tr{\Pi_x}}$ are Kraus operators of the channel $\mathcal{R}^{(x)}$. Indeed, the following calculation shows that $\mathcal{R}^{(x)}$ is trace-preserving:
\begin{align*}
\sum_i V_{xi} V_{xi}^*  + \sum_k K_{xk}^* K_{xk}=\sum_i \Pi_{xi}+\Pi_x=I_\hik.
\end{align*}
Moreover,  $V_{xi}^* \ket{e_{xk}}=0$ for all $i \in \{1, \ldots, n_x\}$, $k \in \{1, \ldots, \tr{\Pi_x}\}$ and $x \in \Omega$, which follows from the definition of $\Pi_x$.
We see that
\begin{align*}
\sum_{x \in \Omega} \mathcal{R}^{(x)}\left( \I_x(\varrho) \right) &= \sum_{x \in \Omega} \sum_{j=1}^{n_x}  V_{xj}^*\left( \sum_{i=1}^{n_x} p_{xi} V_{xi} \varrho V_{xi}^*\right) V_{xj} \\
&= \sum_{x,i} p_{xi}\varrho = id(\varrho).
\end{align*}
Hence, $\I \to id$ so that $\I \leftrightarrow id$.

The second part of the proof is an adaptation of Thm. 2.1 in \cite{NaSe06}
for quantum operations forming a quantum instrument. Let $\I\in \ins(\Omega, \hi, \hik)$ be post-processing equivalent with the identity channel so that there exist channels $\J^{(x)} \in \{\{0 \}, \hik, \hi \}$ such that $\sum_{x \in \Omega} \J^{(x)} \circ \I_x = id$ on $\hi$. Let $\I_x$ and $\J^{(x)}$ have minimal Kraus operators $\{A_{ix}\}_i$ and $\{B^{(x)}_j\}_j$ respectively so that
\begin{equation*}
\varrho = id(\varrho) = \sum_{x,i,j} B^{(x)}_j A_{ix} \varrho A^*_{ix} \left( B^{(x)}_j \right)^*
\end{equation*}
for all $\varrho \in \sh$.

By the unitary equivalence of the Kraus operators it follows that there exists a set of complex numbers $\{u_{ijx}\}_{ijx} \subset \complex$ such that
\begin{align}
\label{eq:def-uijx}
B^{(x)}_j A_{ix} = u_{ijx} I_\hi \quad\quad \forall i,j,x
\end{align}
and $\sum_{i,j,x} |u_{ijx}|^2 =1$. Thus, by multiplying the Eq. (\ref{eq:def-uijx}) by its adjoint on the left, summing over $j$ and noting that $\sum_j \left( B^{(x)}_j \right)^* B^{(x)}_j = I$, we see that $A^*_{i'x} A_{ix} = \beta^x_{i'i} I_{\hi}$ for all $i',i,x$, where we denoted $\beta^x_{i'i} := \sum_{j} u^*_{i'jx} u_{ijx}$. For each $x\in \Omega$ we see that $\beta^x_{i'i}$ is a positive semidefinite matrix which can be diagonalized. Let us denote it's eigenvalues and eigenvectors by $\gamma^x_k$ and $v^x_{ki}$, respectively. Thus, we have $\beta^x_{i'i} =\sum_k \gamma^x_k v^x_{ki'} (v^x_{ki})^*$.
Let us define new set of Kraus operators for each quantum operation $\I_x$ via the relations
\begin{align*}
C_{kx} = \sum_i v^x_{ki} A_{ix}
\end{align*}
Due to unitarity of matrix $\{v^x_{ki}\}_{ik}$ the Kraus operators $\{C_{kx}\}_k$ also represent the quantum operation $\I_x$.
The important property of operators $C_{kx}$ is that their range spaces are orthogonal as is proved via the following calculation:
\begin{align*}
C_{k'x}^* C_{kx} &= \left(\sum_{i'} (v^x_{k'i'})^* A^*_{i'x}\right)\left(\sum_i v^x_{ki} A_{ix}\right)\\
&=\sum_{i',i} (v^x_{k'i'})^*  v^x_{ki}  \left(A^*_{i'x} A_{ix}\right)\\
&=\sum_{i',i} (v^x_{k'i'})^*  v^x_{ki}  \beta^x_{i'i}  I_{\hi}\\
&= \delta_{k'k} \gamma^x_k   I_{\hi},
\end{align*}
where $\delta$ is the Kronecker function. The above equation implies that the singular value decomposition of $C_{kx}$ has the form $C_{kx}=\sqrt{\gamma^x_k} \sum_l \ket{y^{x}_{kl}}\bra{u^x_{kl}}$, where $\ket{u^x_{kl}}$ is an orthonormal basis of $\hi$ and $\ket{y^{x}_{kl}}$ are orthonormal vectors in $\hik$. This means that $C_{kx}$ equals $\sqrt{\gamma^x_k}$ times unitary embedding $V_{xk}\equiv \sum_l \ket{y^{x}_{kl}}\bra{u^x_{kl}}$ of $\hi$ into $\hik$.
Since operations $\I_x$ form an instrument we have that
\begin{align*}
I_{\hi}=\sum_x \sum_k C_{kx}^* C_{kx} = (\sum_x \sum_k \gamma^x_k ) I_{\hi}.
\end{align*}
This allow us to define probability distribution $p_{xi}=\gamma^x_i$ for all $x \in \Omega$ and $i=1,\ldots, n_x$ so that altogether we have
\begin{align*}
\I_x(\varrho) &= \sum^{n_x}_{k=1} C_{kx} \varrho C_{kx}^* =  \sum^{n_x}_{k=1} \gamma^x_k   V_{xk} \varrho V_{xk}^* =\sum_{k=1}^{n_x} p_{xk} V_{xk} \varrho V_{xk}^*,
\end{align*}
which concludes the proof.
\end{proof}

Thus, the equivalence class of the identity channel is the unique greatest element. What about the least element? We can show the following.

\begin{proposition}\label{prop:trash-and-prepare}
Any trash-and-prepare instrument can be post-processed from any instrument. The equivalence class of any trash-and-prepare instrument only consists of trash-and-prepare instruments.
\end{proposition}
\begin{proof}
Let $\I \in \ins(\Omega, \hi, \hik)$ be any instrument and let $\mathcal{T} \in \ins(\Lambda, \hi, \hv)$ be a trash-and-prepare instrument defined as $\mathcal{T}_y(\varrho) = \tr{\varrho} p_y \xi_y$ for all $y \in \Lambda$ for some probability distribution $(p_y)_y$ over $\Lambda$ and some set of states $\{\xi_y\}_y$. To show $\I \to \mathcal{T}$, we see that we can use nearly the same trash-and-prepare instrument as a post-processing, i.e., we set $\R^{(x)}\in \ins(\Lambda, \hik, \hv)$ such that $\R^{(x)}_y(\varrho) = \tr{\varrho} p_y \xi_y$ for all $x \in \Omega$. It follows that
\begin{align*}
\sum_{x \in \Omega} \R^{(x)}_y(\I_x(\varrho)) &= \sum_{x \in \Omega} \tr{\I_x(\varrho)} p_y \xi_y \\
&= \tr{\left(\sum_{x \in \Omega} \I_x \right)(\varrho)} p_y \xi_y \\
&= \tr{\varrho} p_y \xi_y
\end{align*}
for all $y \in \Lambda$ and $\varrho \in \lh$. This proves the first statement.

For the latter part we note that the above applies also for other trash-and-prepare instrument so that any trash-and-prepare instrument can be post-processed from any other trash-and-prepare instrument. To see that this is the whole equivalence class, let $\Q^{(y)} \in \ins(\Gamma, \hv, \hik)$ be any set of post-processing instruments for the formerly defined trash-and-prepare instrument $\mathcal{T}$. Then
\begin{equation*}
\sum_{y \in \Lambda} \Q^{(y)}_z(\mathcal{T}_y(\varrho)) = \tr{\varrho} \sum_{y \in \Lambda} p_y \Q^{(y)}_z(\xi_y) = \tr{\varrho} q_z \sigma_z
\end{equation*}
for all $z \in \Gamma$, where we have defined $q_z = \tr{ \sum_{y \in \Lambda} p_y \Q^{(y)}_z(\xi_y)}$ and $\sigma_z =  \sum_{y \in \Lambda} p_y/q_z \Q^{(y)}_z(\xi_y) \in \sh$ when $q_z \neq 0$ and $\sigma_z = \sigma$ for some fixed $\sigma \in \sh$ when $q_z=0$. Hence, post-processing a trash-and-prepare instrument just leads to another trash-and-prepare instrument.
\end{proof}

We can reformulate the previous result as follows: $\J$ is an instrument such that $\I \to \J$ for all other instruments $\I$ if and only if $\J$ is a trash-and-prepare instrument. Indeed, as in the proof of the previous Proposition, if $\J$ is trash-and-prepare, it can be post-processed from any other instrument $\I$. Conversely, if $\J$ can be post-processed from any instrument then it can be post-processed from some trash-and-prepare instrument from which it follows by the above result that $\J$ also must be trash-and-prepare.

\subsection{Indecomposable instruments}

In the case of POVMs the indecomposable (rank-1) POVMs formed the set of maximal elements of the equivalence classes in the post-processing order. Although we already characterized the single maximal element for instruments, we will see that considering indecomposability in the case of instruments gives us some resemblance to the POVM case.

\begin{definition}
A (nonzero) quantum operation $\mathcal{M}$ is \emph{indecomposable} if $\mathcal{M}= \mathcal{N} + \mathcal{N}'$ for some other quantum operations $\mathcal{N}, \mathcal{N}'$ only when $\mathcal{N} = \mu \mathcal{M}$ and $\mathcal{N}' = \mu' \mathcal{M}$ for some $\mu, \mu' >0$. We call a quantum instrument indecomposable if all of its nonzero operations are indecomposable.
\end{definition}

We can show the following characterization of indecomposable instruments:

\begin{proposition}
A quantum operation is indecomposable if and only if it has (Kraus) rank equal to one.
\end{proposition}
\begin{proof}
First let $\mathcal{N}: \lh \to \lk$ be an indecomposable operation with a minimal Kraus decomposition $\mathcal{N}(\varrho) = \sum_{i=1}^{r} K_i \varrho K^*_i$ for all $\varrho \in \lh$ with Kraus rank $r$. We can define operations $\mathcal{N}_i$ by setting $\mathcal{N}_i(\varrho) = K_i \varrho K_i^*$ for all $\varrho \in \lh$ and $i=1, \ldots, r$, and we see that $\mathcal{N}= \sum_{i=1}^r \mathcal{N}_i$. Because $\mathcal{N}$ is indecomposable, there exists $\nu_i >0$ such that $\mathcal{N}_i = \nu_i \mathcal{N}$ for all $i =1, \ldots, r$ so that the Kraus rank of $\mathcal{N}$ must be one.

Let then $\mathcal{N}: \lh \to \lk$ be an operation with only one Kraus operator, i.e., $\mathcal{N}(\varrho) = K \varrho K^*$ for all $\varrho \in \lh$. Let then $\mathcal{Q}$ and $\mathcal{R}$ be nonzero operations such that $\mathcal{N} = \mathcal{Q} + \mathcal{R}$ with Kraus decompositions $\mathcal{Q}(\varrho)= \sum_i A_i \varrho A^*_i$ and $\mathcal{R}(\varrho) = \sum_j B_j \varrho B^*_j$ for all $\varrho \in \lh$. By the unitary equivalence of the Kraus operators there exists complex numbers $\{u_i\}_i , \{v_j \}_j \subset \complex$ such that $A_i = u_i K $ and $B_j = v_j K$  with $\sum_i |u_i|^2 + \sum_j |v_j|^2 =1$. Hence,
\begin{align*}
\mathcal{Q}(\varrho) &= \sum_i A_i \varrho A_i^* = \left( \sum_i |u_i|^2 \right) K \varrho K^* = u \mathcal{N}(\varrho) \\
\mathcal{R}(\varrho) &= \sum_i B_i \varrho B_i^* = \left( \sum_j |v_j|^2 \right) K \varrho K^* = v \mathcal{N}(\varrho),
\end{align*}
where we have denoted $u := \sum_i |u_i|^2 >0$ and  $v := \sum_j |v_j|^2 >0$. Thus, $\mathcal{N}$ is indecomposable.
\end{proof}

An important class of indecomposable instruments are \textit{the L\"uders instruments}: if $\A \in \obs(\Omega, \hi)$ is a measurement on $\hi$, then the corresponding L\"uders instrument $\I^{\A} \in \ins(\Omega, \hi)$ with induced POVM $\A$ is defined as $\I^{\A}_x(\varrho) = \sqrt{\A(x)} \varrho \sqrt{\A(x)} $ for all $x \in \Omega$ and $\varrho \in \sh$. Since $\sqrt{\A(x)}$ is the only Kraus operator of $\I^{\A}_x$, by the previous characterization of indecomposable instruments we see that L\"uders instruments are indecomposable.

As we saw, just as with POVMs, the indecomposable elements are the ones that have (Kraus) rank equal to one. Although from the previous characterization it is obvious that the indecomposable instruments in general are not maximal elements, we will see that they can be used to produce every instrument as a post-processing of them.

Namely, if $\I \in \ins(\Omega, \hi, \hik)$ has a Kraus decomposition $\I_x(\varrho) = \sum_{i=1}^{n_x} K_{ix} \varrho K^*_{ix}$ for all $x \in \Omega$ and $\varrho \in \sh$ for some $n_x \in \nat$, then it can be (classically) post-processed from the instrument $\hat{\I}$ that is constructed from the single Kraus operators of $\I$, i.e., $\hat{\I}_{(i,x)}(\varrho) = K_{ix} \varrho K^*_{ix}$ for all $i \in \{1, \ldots, n_x\}$, $x \in \Omega$ and $\varrho \in \sh$. By using the (classical) post-processing instruments $\R^{(i,x)}$ defined as $\R^{(i,x)}_{x'} = \delta_{xx'} id_{\hik}$  for all $i \in \{1, \ldots, n_x\}$ and $x,x' \in \Omega$, we see that
\begin{align*}
\sum_{x \in \Omega} \sum_{i=1}^{n_x} \R^{(i,x)}_{x'}(\hat{\I}_{(i,x)}(\varrho)) &= \sum_{x \in \Omega} \sum_{i=1}^{n_x} \delta_{xx'} K_{ix} \varrho K^*_{ix} \\
&= \sum_{i=1}^{n_{x'}} K_{ix'} \varrho K^*_{ix'} \\
&= \I_{x'}(\varrho)
\end{align*}
for all $\varrho \in \sh$ and $x' \in \Omega$ so that $\hat{\I} \to \I$.

We call $\hat{\I}$ the \textit{detailed instrument} of $\I$ and note that any instrument has many detailed instruments depending on their Kraus decomposition, but all the detailed instruments are indecomposable.

We note that it is also known that any instrument $\I$ can be post-processed from the L\"uders instrument that has the same induced POVM as $\I$ \cite{HayashiBook}. Since also L\"uders instruments are indecomposable this is another way to see that every instrument can be post-processed from indecomposable instruments.

To conclude, every instrument can be post-processed from its (indecomposable) detailed instrument. Next, we give a sufficient condition when the inverse statement also holds so that an instrument under this condition is post-processing equivalent to its detailed instrument.

\begin{proposition}\label{prop:detailed-equivalent}
An instrument $\I \in \ins(\Omega, \hi, \hik)$ with Kraus decomposition $\I_x(\varrho) = \sum_{i=1}^{n_x} K_{ix} \varrho K_{ix}^*$ is post-processing equivalent with its detailed instrument if $K_{ix}^* K_{jx} = 0$ for all $i \neq j$ and $x \in \Omega$.
\end{proposition}
\begin{proof}
We have shown above that $\hat{\I} \to \I$ for any instrument $\I$. Next we focus on proving the relation $\I \to \hat{\I}$ under the orthogonality condition of the Kraus operators of $\I_x$ for all $x \in \Omega$.

Suppose $K_{ix}^* K_{jx}=0$ for all $i \neq j$ for all $x \in \Omega$. If $\pi_{ix}$ is the projector onto the image of $K_{ix}$ we see that together they are mutually orthogonal for all $x\in \Omega$. Define projectors $\{\Pi_{ix}\}$ as $\Pi_{1x}=I-\sum_{i\neq 1} \pi_{ix}$, $\Pi_{ix}=\pi_{ix}$ for $i\neq 1$ and all $x \in \Omega$.
Define instruments $\R^{(x)} \in \ins(\{1, \ldots, n_x\} \times \Omega, \hi, \hik)$ as $\R^{(x)}_{(i,y)}(\varrho)= \delta_{xy}\Pi_{iy} \varrho \Pi_{iy}$. Now we see that for $i\neq 1$, we have that
\begin{align*}
\sum_{x \in \Omega} \R^{(x)}_{(i,y)}(\I_x(\varrho)) = \sum_{x \in \Omega}\sum_{j=1}^{n_x} \delta_{xy} \pi_{iy} K_{jx} \varrho  K_{jx}^* \pi_{iy} = K_{iy} \varrho K_{iy},
\end{align*}
and for $i=1$ we have
\begin{align*}
\sum_{x \in \Omega} \R^{(x)}_{(1,y)}(\I_x(\varrho)) &= \sum_{x \in \Omega} \sum_{j=1}^{n_x} \left[ \delta_{xy}
 \left(I -\sum_{i\neq 1} \pi_{iy} \right) K_{jx} \varrho K_{jx}^* \right. \\
 & \left. \ \ \ \times \left(I -\sum_{i'\neq 1} \pi_{i'y} \right)\right] \\
&= \sum_{j=1}^{n_y} K_{jy} \varrho K_{jy}^* - \sum_{j=1}^{n_y} \sum_{i \neq 1}  \pi_{iy} K_{jy} \varrho  K_{jy}^*\\
& \ \ \   - \sum_{j=1}^{n_y} \sum_{i' \neq 1}  K_{jy} \varrho  K_{jy}^* \pi_{i'y} \\
& \ \ \ + \sum_{j=1}^{n_y} \sum_{i,i' \neq 1}  \pi_{iy} K_{jy} \varrho  K_{jy}^* \pi_{i'y} \\
&= \sum_j K_{jy} \varrho K_{jy}^* - \sum_{j \neq 1} K_{jy} \varrho K_{jy}^* \\
& \ \ \  - \sum_{j \neq 1} K_{jy} \varrho K_{jy}^*
+ \sum_{j \neq 1} K_{jy} \varrho K_{jy}^* \\
 &= K_{1y} \varrho K^*_{1y},
\end{align*}
so that $\sum_{x \in \Omega} \R^{(x)}_{(i,y)}(\I_x(\varrho)) = K_{iy} \varrho K^*_{iy} = \hat{\I}_{(i,y)}(\varrho)$ for all $i \in \{1, \ldots,n_y\}$ and $y \in \Omega$ for all $\varrho \in \sh$. Hence, $\I \to \hat{\I}$.
\end{proof}

For measure-and-prepare instruments we can use the previous result to show a necessary and sufficient condition for being equivalent with an indecomposable instrument.

\begin{proposition}
A measure-and-prepare instrument $\I \in \ins(\Omega, \hi, \hik)$ is post-processing equivalent with an indecomposable instrument if and only if $\A^\I \in \obs(\Omega, \hi)$ is indecomposable.
\end{proposition}
\begin{proof}
Let us consider a measure-and-prepare instrument $\I \in \ins(\Omega, \hi, \hik)$ that is of the form $\I_x(\varrho) =  \tr{\A(x) \varrho} \xi_x$ for all $x \in \Omega$ and $\varrho \in \sh$ for some POVM $\A \in \obs(\Omega, \hi)$ and some states $\{\xi_x\}_{x \in \Omega} \subset \shik$.  For each $x \in \Omega$, let us consider the spectral decomposition of the state $\xi_x = \sum_{i=1}^{n_x} p_{ix} \kb{\varphi_{ix}}{\varphi_{ix}}$, where $p_{ix} \geq 0$, $\sum_i p_{ix} =1$ and $\{\varphi_{ix}\}_i$ is a set of orthonormal vectors in $\hik$ for all $x \in \Omega$. Similarly for each $x \in \Omega$, we can write $\A(x)$ as $\A(x) = \sum_{j=1}^{m_x} q_{jx} \kb{\psi_{jx}}{\psi_{jx}}$ for some orthogonal set of vectors $\{\psi_{jx}\}_j$ in $\hi$ and some positive numbers $\{q_{jx}\}_j$. Let us define operators $K_{ijx} := \sqrt{p_{ix} q_{jx}} \kb{\varphi_{ix}}{\psi_{jx}}$ for all $i \in \{1, \ldots, n_x\}$ and $j \in \{1, \ldots, m_x\}$ for each $x \in \Omega$. One can confirm that $\sum_{i,j} K^*_{ijx} K_{ijx} = \A(x)$ for all $x \in \Omega$ so that $\sum_{i,j,x} K^*_{ijx} K_{ijx} = I_\hi$, and that $\sum_{i=1}^{n_x}\sum_{j=1}^{m_x} K_{ijx} \varrho K^*_{ijx} = \I_x(\varrho)$ for all $x \in \Omega$ so that $\{K_{ijx}\}_{i,j}$ is a set of Kraus operators for $\I_x$.

Let first $\A$ be indecomposable, i.e. rank-1, so that $m_x=1$ for all $x \in \Omega$ and we can omit the index $j$ in the previous consideration and thus $K_{ix} := \sqrt{p_{ix} q_{1x}} \kb{\varphi_{ix}}{\psi_{1x}}$ form the set of Kraus operators for $\I_x$ for all $x \in \Omega$. We see that $K^*_{ix}K_{i'x} = 0$ for all $i \neq i'$ for all $x \in \Omega$ and thus by Prop. \ref{prop:detailed-equivalent} the instrument $\I$ is equivalent with its detailed instrument related to that Kraus decomposition. From the proof of Prop. \ref{prop:detailed-equivalent} we see that the simulator instruments $\R^{(x)} \in \ins(\{1, \ldots, n_x\}\times\Omega, \hik)$ can be written as $\R^{(x)}_{(i,y)}(\varrho) = \delta_{xy} \Pi_{ix} \varrho \Pi_{ix}$, where the projectors $\{\Pi_{ix}\}_i$ are defined as $\Pi_{ix} = \kb{\varphi_{ix}}{\varphi_{ix}}$ for all $i \neq 1$ and $\Pi_{1x} = I_\hik - \sum_{i \neq 1} \kb{\varphi_{ix}}{\varphi_{ix}}$ for all $x \in \Omega$.

Let then $\A$ not be rank-1, i.e., there exists $x' \in \Omega$ such that $m_{x'} \geq 2$. Suppose that $\I$ is equivalent to some indecomposable instrument $\J \in \ins(\Lambda, \hi, \hv)$ with Kraus decomposition $\J_y(\varrho) = L_y \varrho L^*_y$ for all $y \in \Lambda$ and $\varrho \in \sh$. Thus, there exist instruments $\R^{(x)} \in \ins(\Lambda, \hik, \hv)$ with Kraus operators $R^{(x)}_{ky}$ such that
\begin{align*}
 L_y \varrho L^*_y &= \J_y(\varrho) = \sum_{x \in \Omega} \R^{(x)}_y(\I_x(\varrho)) \\
 &= \sum_{i,j,k,x} R^{(x)}_{ky} K_{ijx} \varrho K^*_{ijx} \left( R^{(x)}_{ky} \right)^*
\end{align*}
for all $y \in \Lambda$ and $\varrho \in \sh$. From the unitary equivalence of the Kraus operators it follows that there exist complex numbers $u_{ijkxy} \in \complex$ such that
\begin{equation}\label{eq:m-a-p-kraus-equivalence}
R^{(x)}_{ky} K_{ijx} = u_{ijkxy} L_y
\end{equation}
for all $i,j,k,x,y$ and $\sum_{i,j,k,x} |u_{i,j,k,x}|^2 = 1$ for all $y \in \Lambda$. If we denote the induced POVM of $\R^{(x)}$ by $\mathsf{R}^{(x)}$, i.e., $\mathsf{R}^{(x)}(y) = \sum_k \left( R^{(x)}_{ky} \right)^* R^{(x)}_{ky}$ for all $y \in \Lambda$ for all $x \in \Omega$, by multiplying Eq. \eqref{eq:m-a-p-kraus-equivalence} by its adjoint from the left and by summing over indices $i,k$, we see that
\begin{align*}
\left( \sum_{i,k} |u_{ijkxy}|^2 \right) L^*_y L_y &= \sum_{i,k} K^*_{ijx} \left(R^{(x)}_{ky} \right)^* R^{(x)}_{ky} K_{ijx} \\
&= \sum_i  p_{ix} \ip{\varphi_{ix}}{ \sum_k \left( R^{(x)}_{ky} \right)^* R^{(x)}_{ky} \varphi_{ix}}  \\
& \ \ \ \ \ \ \times   q_{jx} \kb{\psi_{jx}}{\psi_{jx}} \\
&= \tr{\mathsf{R}^{(x)}(y) \xi_x} q_{jx} \kb{\psi_{jx}}{\psi_{jx}}
\end{align*}
for all $j \in \{1, \ldots, m_x\}$, $x \in \Omega$ and $y\in \Lambda$.

As we mentioned, since $\A$ is not rank-1, there exists $x' \in \Omega$ such that $m_{x'} \geq 2$ so that there are indices $j',j'' \in \{1, \ldots, m_{x'}\}$ such that $q_{j'x'}, q_{j''x'} >0$. Also, since $\mathsf{R}^{(x')}$ is a POVM, there exists $y' \in \Lambda$ such that $\tr{\mathsf{R}^{(x')}(y') \xi_{x'}} \neq 0$. Thus, we have that
\begin{align*}
\left( \sum_{i,k} |u_{ij'kx'y'}|^2 \right) L^*_{y'} L_{y'} &= \tr{\mathsf{R}^{(x')}(y') \xi_{x'}} q_{j'x'} \kb{\psi_{j'x'}}{\psi_{j'x'}} \\
\left( \sum_{i,k} |u_{ij''kx'y'}|^2 \right) L^*_{y'} L_{y'} &= \tr{\mathsf{R}^{(x')}(y') \xi_{x'}} q_{j''x'} \kb{\psi_{j''x'}}{\psi_{j''x'}}
\end{align*}
which leads to a contradiction since $\kb{\psi_{j'x'}}{\psi_{j'x'}}$ is not proportional to $\kb{\psi_{j''x'}}{\psi_{j''x'}}$. Hence, if $\A$ is not rank-1, then $\I$ is not post-processing equivalent to any indecomposable instrument.
\end{proof}

This result shows that there are instruments that are not equivalent with an indecomposable instruments and this is the case for all trash-and-prepare instruments. On the other hand, it is known that if the induced POVM $\A^\I$ of any instrument $\I$ is rank-1, then $\I$ must be a measure-and-prepare instrument \cite{HeWo10}. Hence, we get the following corollary.

\begin{corollary}
Any instrument with an indecomposable induced POVM is post-processing equivalent with its detailed (indecomposable) instrument.
\end{corollary}

\section{Post-processing of the induced POVMs}\label{sec:induced-POVM}

So far we have separately considered the post-processing of POVMs and instruments. However, as quantum instruments are also a form of measurement, we can start finding connections between the two post-processings. When the resulting instruments is indecomposable, we can show the following:

\begin{proposition}\label{prop:associated-POVMs}
Let instruments $\I \in \ins(\Omega, \hi, \hik)$ and $\J \in \ins(\Lambda, \hi, \hv)$ be such that $\J$ is indecomposable. If $\I \to \J$, then $\A^\J \to \A^\I$.
\end{proposition}
\begin{proof}
If $\I \to \J$, then there exists post-processings $\R^{(x)} \in \ins(\Lambda, \hik, \hv)$ for all $x \in \Omega$ such that $\J_y = \sum_{x \in \Omega} \R^{(x)}_y \circ \I_x$ for all $y \in \Lambda$. Let $\{A_{ix}\}_i$, $B_y$ and $\{R^{(x)}_{ky}\}_k$ be Kraus operators for $\I_x$, $\J_y$ and $\R^{(x)}_y$ respectively. Thus,
\begin{equation*}
B_{y} \varrho B^*_{y} = \sum_{i,k,x} R^{(x)}_{ky} A_{ix} \varrho A^*_{ix} \left( R^{(x)}_{ky}\right)^*
\end{equation*}
for all $\varrho \in \sh$.

From the unitary equivalence of the Kraus operators it follows that there exists complex numbers $\{u_{ikxy}\}_{i,k,x,y} \subset \complex$ such that $R^{(x)}_{ky} A_{ix} =  u_{ikxy} B_{y}$ for all $i,k,x,y$ and $\sum_{i,k,x} |u_{ikxy}|^2 =1$ for all $y \in \Lambda$. By multiplying the left side of the previous expression by its adjoint and summing over the indices $i,k,y$, we see that
\begin{equation*}
\sum_i A^*_{ix} A_{ix} = \sum_{i,k,y} A^*_{ix} \left( R^{(x)}_{ky}\right)^* R^{(x)}_{ky} A_{ix} = \sum_{i,k,y} |u_{ikxy}|^2 B^*_y B_y
\end{equation*}
for all $x \in \Omega$.

Clearly $\A^\I(x) = \sum_i A^*_{ix}A_{ix}$ and $\A^\J(y)= B^*_y B_y$ for all $x \in \Omega$ and $y \in \Lambda$. If we denote $\nu_{yx} = \sum_{i,k} |u_{ikxy}|^2 \geq 0$, we see that $\sum_{x \in \Omega} \nu_{yx} = 1$ for all $y \in \Lambda$ so that it defines a post-processing $(\nu_{yx})_{y\in \Lambda, x \in \Omega}$ such that $\A^\I(x) = \sum_{y \in \Lambda} \nu_{yx} \A^\J(y)$ for all $x \in \Omega$. Hence, $\A^\J \to \A^\I$.
\end{proof}

We note that the previous claim is not true when $\J$ is not indecomposable. To see this, let us consider the case when $\J$ is a trash-and-prepare instrument. As we have shown in Prop. \ref{prop:trash-and-prepare}, then $\I \to \J$ for any instrument $\I$, but since $\A^\J$ is now a trivial POVM the relation $\A^\J \to \A^\I$ does not hold in general.

When both instruments are indecomposable, we can prove even a stronger result, but for that we need a small lemma first.

\begin{lemma}
\label{lemma:KLrel}
Suppose that operators $K:\hi \rightarrow \hik,L: \hi\rightarrow \hv$ satisfy $K^* K= c L^* L$ for some $c>0$. Then there exists an operator $U: \hv \rightarrow \hik$ such that $K=\sqrt{c} U L$ with the following properties.
\begin{enumerate}
  \item If $\dim \hik \geq \dim\hv$, then $U$ is an isometry.
  \item If $\dim \hik < \dim \hv$, then $U$ is a partial isometry such that its range is the whole $\hik$ and $U^* U L=L$.
\end{enumerate}
\end{lemma}

\begin{proof}
Let us start with a small note about partial isometries. Suppose we have an operator $X:\hi \rightarrow \hik $ defined as $X=\sum_{m=a}^b \ket{e_m}\bra{g_m}$, where $\{\ket{g_m}\}_{m=a}^{b}$ and $\{\ket{e_m}\}_{m=a}^{b}$ are two sets of orthonornal vectors from $\hi$ and $\hik$, respectively.
Thus, operator $X$ isometrically transfers subspace $V_{in}=\mathrm{span}\left(\{\ket{g_m}\}_{m=a}^{b}\right)$ into subspace $V_{out}=\mathrm{span}\left(\{\ket{e_m}\}_{m=a}^{b}\right)$. Suppose $\{\ket{h_m}\}_{m=a}^{b}$ is another set of orthonormal vectors that span $V_{in}$. Consequently, a projector onto $V_{in}$
can be written as $P=\sum_{m=a}^b \ket{g_m}\bra{g_m}=\sum_{m=a}^b \ket{h_m}\bra{h_m}$.
Clearly,
\begin{align}
\label{eq:rewritePI}
X&=XP=\sum_{m=a}^b \ket{e_m}\bra{g_m}\sum_{n=a}^b \ket{h_n}\bra{h_n} =\sum_{n=a}^b \ket{\tilde{e}_n}\bra{h_n},
\end{align}
where vectors $\ket{\tilde{e}_n}=\sum_{m=a}^b \ip{g_m}{h_n}\ket{e_m}$ are orthonormal as one can easily check.
Thus, we see that $X$ can be also seen as a (linear) isometric transformation of orthonormal vectors $\ket{h_n}$ onto orthonormal vectors $\ket{\tilde{e}_n}$.

Next, we consider singular value decompositions of operators $K$ and $L$
\begin{align} \label{eq:defkl}
K=\sum_{m=1}^k \lambda_m \ket{e_m}\bra{g_m}
\quad
L=\sum_{n=1}^l \mu_n\ket{f_n}\bra{h_n},
\end{align}
where the singular values $\lambda_m,\mu_n>0$ are arranged in the decreasing order and  $\{\ket{g_m}\}_{m=1}^{k}$, $\{\ket{e_m}\}_{m=1}^{k}$, $\{\ket{h_n}\}_{n=1}^{l}$,$\{\ket{f_n}\}_{n=1}^{l}$ are orthonormal vectors in the corresponding Hilbert spaces and we assume they were extended to form an orthonormal basis in each of the spaces.
Equality $K^* K= c L^* L$ can be now written as
\begin{align*}
\sum_{m=1}^k \lambda_m^2 \ket{g_m}\bra{g_m}=\sum_{n=1}^l c \mu_n^2 \ket{h_n}\bra{h_n}.
\end{align*}

Both left and right side have form of a spectral decomposition for the same positive-semidefinite operator $K^*K$. This has important consequences for the singular value decompositions (\ref{eq:defkl}). First of all, $k=l$ and $\lambda_m = \sqrt{c} \mu_m$. If some of the eigenvalues of $K^*K$ (or equivalently singular values of $K$ or $L$) are degenerate then for every such eigenspace defined by eigenvalue $\lambda_a =\lambda_{a+1}=\ldots = \lambda_b$ we have that
\begin{align*}
P_{\lambda_a}=\sum_{m=a}^b \ket{g_m}\bra{g_m}=\sum_{m=a}^b \ket{h_m}\bra{h_m},
\end{align*}
i.e. both $\{\ket{g_m}\}_{m=a}^{b}$ and $\{\ket{h_n}\}_{n=a}^{b}$ are orthonormal basis of this eigenspace.
Using the considerations about partial isometries (specifically Eq. (\ref{eq:rewritePI})) from the beginning of this proof we can rewrite operator $K$ as
\begin{align}\label{eq:rewriteK}
K=\sum_{m=1}^k \sqrt{c} \mu_m \ket{\tilde{e}_m}\bra{h_m}.
\end{align}
We extend orthonormal vectors $\{\ket{\tilde{e}_m}\}_{m=1}^{k}$ to form an orthonormal basis of $\hik$.
At his point we have to consider separately two cases based on the relation between dimensions of $\hik$ and $\hv$.
First, let us consider $\dim \hik \geq \dim\hv$.
We note that $k\leq \dim \hv \leq \dim \hik$ and before we defined
$\{\ket{f_n}\}_{n=1}^{\dim \hv}$ as a complete orthonormal basis of $\hv$.
We can now define operator $U:\hv\rightarrow \hik$ as
\begin{align}\label{eq:defU1}
U=\sum_{m=1}^{\dim \hv}\ket{\tilde{e}_m}\bra{f_m}.
\end{align}
Clearly, $U$ is an isometry by definition and direct calculation verifies that $\sqrt{c}\; UL$ equals $K$ as given in Eq. (\ref{eq:rewriteK}).
Next, we consider $\dim \hik < \dim \hv$.  Let us now set
\begin{align*}
U=\sum_{m=1}^{\dim \hik}\ket{\tilde{e}_m}\bra{f_m}.
\end{align*}
As we see from the above definition, $U$ is a partial isometry, whose range is the whole $\hik$. Due to $k\leq \dim \hik$ we again have that $\sqrt{c}\; UL=K$ and one can also easily verify that $U^* U L=L$, which concludes the proof.
\end{proof}

We can now show the following:

\begin{proposition}\label{prop:rank1equivalenceGeneral}
Indecomposable instruments $\I \in \ins(\Omega, \hi, \hik)$ and $\J \in \ins(\Lambda, \hi, \hv)$ are equivalent ($\I \leftrightarrow \J$) if and only if their induced POVMs are equivalent ($\A^\I \leftrightarrow \A^\J$).
\end{proposition}
\begin{proof}
Let us assume that the indecomposable instruments $\I$ and $\J$ are equivalent, i.e. both $\I \rightarrow \J$ and $\I \leftarrow \J$ holds.
Using Proposition \ref{prop:associated-POVMs} we get that both $\A^\J \to \A^\I$ and $\A^\I \to \A^\J$ hold, respectively, which means that $\A^\I \leftrightarrow \A^\J$.

For the opposite direction we assume that $\A^\I \leftrightarrow \A^\J$. Without loss of generality we may assume $\dim \hik \geq \dim \hv$ and that $\A^\J$ and $\A^\I$ are non-vanishing so that they only consist of non-zero effects. Our first goal is to show that under the assumption $\A^\I \leftrightarrow \A^\J$, from the relation $\A^\I \to \A^\J$ we get $\I \rightarrow \J$ for indecomposable instruments.

From Prop. \ref{prop:pp-equivalent-POVMs} we have that there is a post-processing matrix $\nu$ such that
\begin{equation}
\label{eq:defpyx}
\A^\J(y) = \sum_{x \in \Omega} \nu_{xy} \A^\I(x)
\end{equation}
 for all $y \in \Lambda$ such that $\nu_{xy} \neq 0$ only if $\A^\I(x)$ is proportional to $\A^\J(y)$ for all $x \in \Omega$ and $y \in \Lambda$. Thus, for all $\nu_{xy}>0$ there exists $c_{xy}>0$ such that
\begin{equation}
\label{eq:defcxy}
\A^\I(x)=c_{xy}\A^\J(y).
\end{equation}

Inserting Eq. (\ref{eq:defcxy}) into Eq. (\ref{eq:defpyx}) we get
$\A^\J(y) = \sum_{x: \nu_{xy}>0}\; \nu_{xy} c_{xy}\;\A^\J(y)$ or equivalently
\begin{align}
\label{eq:relpyxcxy}
\sum_{x: \nu_{xy} >0} \nu_{xy} c_{xy}=1
\end{align}
for all $y \in \Lambda$.

If we denote the Kraus operators of the instruments as
\begin{align}
\label{eq:defkrausIJ}
\I_x (\rho)&= A_{x} \rho A^*_{x} \quad\quad \J_y (\rho)= B_{y} \rho B^*_{y}
\end{align}
then Eq. (\ref{eq:defcxy}) can be rewritten as
\begin{align*}
A^*_{x} A_{x} &= c_{xy} B^*_{y} B_{y}
\end{align*}
Using lemma \ref{lemma:KLrel} we obtain
\begin{align*}
A_x=\sqrt{c_{xy}} U_{xy} B_y,
\end{align*}
where $U_{xy}^* U_{xy}=I_\hv$. On the other hand, $\Pi_{xy}\equiv U_{xy} U_{xy}^*$ can be a nontrivial projector on $\hik$. We denote its complement as $\overline{\Pi}_{xy}=I_\hik - \Pi_{xy}$ and we define orthonormal states $\{\ket{e^{xy}_k}\} \subset  \hik$ such that $\overline{\Pi}_{xy}=\sum_{k=1}^{m_{xy}} \ket{e^{xy}_k}\bra{e^{xy}_k}$ with $m_{xy}=\tr{\overline{\Pi}_{xy}}$.
We note that by construction $U_{xy}^* \ket{e^{xy}_k}=0$ for all $k \in \{1, \ldots, m_{xy}\}$, $x \in \Omega$ and $y \in \Lambda$.
For every $x\in \Omega$ we define instrument $\R^{(x)} \in \ins(\Lambda, \hik, \hv)$ via the following formula
\begin{align*}
\R^{(x)}_y (\rho)= R^{(x)}_y \rho (R^{(x)}_y)^* + \sum_{k=1}^{m_{xy}} Q^{(x)}_{ky} \rho (Q^{(x)}_{ky})^* ,
\end{align*}
where $R^{(x)}_y=\sqrt{\nu_{xy}} U_{xy}^*$ and $Q^{(x)}_{ky}=\sqrt{\nu_{xy}} \ket{\xi}\bra{e^{xy}_k}$ for some fixed unit vector $\ket{\xi} \in \hv$.
Complete positivity of the instrument is obvious from its definition and we check preservation of the trace via the following calculation
\begin{align*}
\sum_{y\in\Lambda} &\left((R^{(x)}_y)^*R^{(x)}_y  + \sum_{k=1}^{m_{xy}} (Q^{(x)}_{ky})^*Q^{(x)}_{ky}\right)= \nonumber \\
&=\sum_{y\in\Lambda} \nu_{xy} \left( U_{xy} U_{xy}^* + \sum_{k=1}^{m_{xy}} \ket{e^{xy}_k}\bra{e^{xy}_k}\right) \nonumber \\
&=\sum_{y\in\Lambda} \nu_{xy} \left( \Pi_{xy} + \overline{\Pi}_{xy}\right)=I_{\hik}.
\end{align*}
Let's now evaluate the post-processing of instrument $\I$ via the instruments $\R^{(x)}$.
We obtain
\begin{align*}
\sum_{x\in\Omega} \R^{(x)}_y(\I_x(\rho))&= \sum_{x\in\Omega}\left(R^{(x)}_y A_x \rho A^*_x (R^{(x)}_y)^* \right. \\
& \ \ \ \left.+ \sum_{k=1}^{m_{xy}} Q^{(x)}_{ky} A_x \rho A^*_x (Q^{(x)}_{ky})^*\right) \\
&= \sum_{x\in\Omega }c_{xy}\nu_{xy} \left[  U_{xy}^*  U_{xy} B_y \rho B^*_y U^*_{xy} U_{xy} \right. \\
& \quad \left.  + \sum_{k=1}^{m_{xy}} \bra{e^{xy}_k}  U_{xy} B_y \rho  B^*_y U^*_{xy}\ket{e^{xy}_k}  \ket{\xi}\bra{\xi} \right] \nonumber \\
&= \sum_{x: \nu_{xy} >0} \left(c_{xy} \nu_{xy}\right)  B_y \rho B^*_y   \nonumber \\
&= B_y \rho B^*_y = \J_y(\rho),
\end{align*}
where we used $U_{xy}^* \ket{e^{xy}_k}=0$, $U_{xy}^*  U_{xy}=I_{\hv}$ and Eqs. (\ref{eq:relpyxcxy}) and (\ref{eq:defkrausIJ}).
Thus, when  $\A^\J \leftrightarrow \A^\I$, we have proved that $\A^\I \to \A^\J$ implies  $\I \rightarrow \J$.

Our next goal is to prove that the equivalence $\A^\J \leftrightarrow \A^\I$ for indecomposable instruments implies also $\J \rightarrow \I$.
We start by explicitly writing out the equations guaranteed by Prop. \ref{prop:pp-equivalent-POVMs} for $\A^\J \to \A^\I$. Thus, there exists a stochastic matrix $\mu$ such that
\begin{equation}
\label{eq:defmuyx}
\A^\I(x) = \sum_{y \in \Lambda} \mu_{yx} \A^\J(y)
\end{equation}
for all $x\in \Omega$ such that $\mu_{yx} \neq 0$ only if $\A^\J(y)$ is proportional to $\A^\I(x)$. Thus, for all $\mu_{yx}>0$ we have that there exists $d_{yx}>0$ such that
\begin{equation}
\label{eq:defdyx}
\A^\J(y)=d_{yx}\A^\I(x).
\end{equation}

Inserting Eq. (\ref{eq:defdyx}) into Eq. (\ref{eq:defmuyx}) we get
$\A^\I(x) = \sum_{y: \mu_{yx}>0}\; \mu_{yx} d_{yx}\;\A^\I(x)$ or equivalently
\begin{align}
\label{eq:relmuyxdxy}
\sum_{y: \mu_{yx}>0}\; \mu_{yx} d_{yx}=1
\end{align}
for all $x \in \Omega$.

Using Kraus operators $A_{x}$, $B_{y}$ for $\I_x$ and $\J_y$ respectively, Eq. (\ref{eq:defdyx}) can be rewritten as
\begin{align*}
B^*_{y} B_{y}= d_{yx} A^*_{x} A_{x}.
\end{align*}
Applying Lemma \ref{lemma:KLrel} to this equation, we obtain
\begin{align*}
B_y=\sqrt{d_{yx}} V_{yx} A_x,
\end{align*}
where $V_{yx}$ is a partial isometry since $\dim \hv \leq \dim \hik$. On one hand we have that $V_{yx} V_{yx}^*=I_{\hv}$, and on the other hand as a consequence of Lemma \ref{lemma:KLrel} we have
\begin{align}
\label{eq:conslemma1}
V^*_{yx} V_{yx} A_x = A_x
\end{align}
for all $x \in \Omega$ and $y \in \Lambda$.
For every $y\in \Lambda$ we define instrument $\G^{(y)} \in \ins(\Omega, \hv, \hik)$ via the following formula
\begin{align*}
\G^{(y)}_x (\rho)= G^{(y)}_x \rho (G^{(y)}_x)^*,
\end{align*}
where $G^{(y)}_x=\sqrt{\mu_{yx}} V_{yx}^*$.
We check the trace preservation by evaluating
\begin{align*}
\sum_{x\in\Omega} (G^{(y)}_x)^* G^{(y)}_x = \sum_{x\in\Omega} \mu_{yx} V_{yx}  V_{yx}^* = \sum_{x\in\Omega} \mu_{yx} I_{\hv}= I_{\hv}.
\end{align*}
Finally we calculate the post-processing of instrument $\J$ via the instruments $\G^{(y)}$.
We obtain
\begin{align*}
\sum_{y\in\Lambda} \G^{(y)}_x(\J_y(\rho))&= \sum_{y\in\Lambda} G^{(y)}_x B_y \rho B^*_y (G^{(y)}_x)^* \\
& =\sum_{y\in\Lambda} \mu_{yx} d_{yx} V_{yx}^*  V_{yx} A_x\rho A_x^* V^*_{yx} V_{yx} \nonumber\\
&=\sum_{y: \mu_{yx}>0}\; (\mu_{yx} d_{yx}) A_x\rho A_x^* \\
&=A_x\rho A_x^* = \I_x(\rho),
\end{align*}
where we used Eqs. (\ref{eq:conslemma1}), (\ref{eq:relmuyxdxy}) and (\ref{eq:defkrausIJ}). This concludes the proof, since we showed that $\A^\J \leftrightarrow \A^\I$ implies $\I \leftrightarrow \J$ for indecomposable instruments.
\end{proof}

From the previous result we see that in the special case when we have two L\"uder's instruments $\I^\A$ and $\I^\B$ for two POVMs $\A$ and $\B$, then $\I^\A$ and $\I^\B$ are equivalent if and only if $\A$ and $\B$ are equivalent. 

For measure-and-prepare instruments we already saw in Example \ref{ex:m-a-p} that if $\I \to \J$ and $\I$ is a measure-and-prepare instrument, then $\A^\I \to \A^\J$. We can now show that this is actually both necessary and sufficient condition in the case when also $\J$ is a measure-and-prepare instrument.

\begin{proposition}
Let  $\I \in \ins(\Omega, \hi, \hik)$ and $\J \in \ins(\Lambda, \hi, \hv)$ be measure-and-prepare instruments. Then $\I \to \J$ if and only if $\A^\I \to \A^\J$.
\end{proposition}
\begin{proof}
Since $\I$ and $\J$ are measure-and-prepare, then there exists POVMs $\A \in \obs(\Omega, \hi)$ and $\B \in \obs(\Lambda,\hi)$ as well as states $\{\sigma_x\}_{x \in \Omega}\subset \shik$  and $\{\xi_y\}_{y \in \Lambda} \subset \shv$ such that
\begin{align*}
\I_x(\varrho) &= \tr{\A(x) \varrho} \sigma_x, \\
\J_y(\varrho) &= \tr{\B(y) \varrho} \xi_y
\end{align*}
for all $x \in \Omega$ and $y \in \Lambda$. We see that $\A^\I = \A$ and $\A^\J = \B$. By following the steps of Example \ref{ex:m-a-p}, we see that if $\I \to \J$, then $\A = \A^\I \to \A^\J = \B$.

Now let $\A^\I \to \A^\J$, i.e., $\A \to \B$ so that $\B(y) = \sum_{x \in \Omega} \mu_{xy} \A(x)$ for all $y \in \Lambda$ for some post-processing $\mu = (\mu_{xy})_{x \in \Omega, y \in \Lambda}$. Let us define instruments $\R^{(x)} \in \ins(\Lambda, \hik, \hv)$ by setting $\R^{(x)}_y(\varrho) = \mu_{xy} \xi_y$ for all $x \in \Omega$, $y \in \Lambda$ and $\varrho \in \shik$. We now see that
\begin{align*}
\sum_{x \in \Omega} \R^{(x)}_y(\I_x(\varrho)) &=\sum_{x \in \Omega} \tr{\A(x) \varrho} \R^{(x)}_y(\sigma_x) \\
&= \tr{\sum_{x \in \Omega} \mu_{xy} \A(x) \varrho} \xi_y \\
&= \tr{\B(y)\varrho} \xi_y = \J_y(\varrho)
\end{align*}
for all $y \in \Lambda$. Hence, $\I \to \J$.
\end{proof}

\section{Simulation of instruments}\label{sec:simulation}
The simulation scheme describes a process of obtaining new devices out of some existing devices by the means of operational manipulations. For example, in the case of measurements (see \cite{GuBaCuAc17, OsGuWiAc17, FiHeLe18, OsMaPu19}), from a set of measurement devices new observables can be obtained by classical means of mixing and/or post-processing the classical outcomes. This can be achieved by giving probability for each device according to which we use it in a measurement and/or by (classically) post-processing the obtained measurement outcomes.

The previously described procedure can be used to simulate a measurement device that is not directly at hand and which may be hard to implement by itself. One can consider which observables can be obtained from a single POVM via post-processing, which effectively characterizes joint measurability \cite{HeMiZi16}, and in this sense the above concept of measurement simulability can be considered as a generalization of joint measurability. One can also ask, which POVMs are needed to simulate all observables (simulation irreducible measurements \cite{FiHeLe18}), or what can one get out of a given observables with specific properties (projective measurements \cite{OsGuWiAc17, OsMaPu19}, effectively dichotomic measurements \cite{KlCa16, KlVeCa17, Huetal18, FiGuHeLe20}).

Next, we will define simulation of instruments analogously to the simulation of POVMs by using the post-processing that was defined and studied in the previous sections. We start by briefly recalling mixing of instruments.

\subsection{Mixing of quantum instruments}

For a fixed outcome set $\Omega$ and Hilbert spaces $\hi$ and $\hik$ the set of instruments $\ins(\Omega, \hi, \hik)$ is convex. Namely, if we have devices described by instruments $\{\I^{(i)}\}_{i=1}^n \subset \ins(\Omega, \hi, \hik)$, then we can choose to use device $\I^{(i)}$ with probability $p_i$ in every round of the experiment with some probability distribution $(p_i)_{i=1}^n$. The new instrument $\I$ that is formed as a mixture is then defined as
\begin{equation*}
\I_x = \sum_{i=1}^n p_i \I^{(i)}_x
\end{equation*}
for all $x \in \Omega$. Note that we can always consider instruments to have the same outcome sets by just adding zero outcomes to instruments if needed, but they still have to have the same input and output spaces $\lh$ and $\lk$.

We can also consider other type of mixing where we also keep track of the instrument that was used in each round of the experiment. Then we consider the mixed instrument to have to outcomes, first outcome indicating the instrument that was used and the second giving the outcome that was obtained from the instrument that was used. Thus, in this case we define the new instrument $\tilde{\I}$ to have an outcome set $\{1, \ldots, n\}\times  \Omega$ so that
\begin{equation*}
\tilde{\I}_{(i,x)} = p_i \I^{(i)}_x
\end{equation*}
for all $i \in \{1, \dots, n\}$ and $x \in \Omega$. We note that the traditional mixture $\I$ where we do not keep track of the measured instrument can be obtained as a post-processing of the instrument $\tilde{\I}$, namely $\I_x = \sum_i \tilde{\I}_{(i,x)}$ for all $x \in \Omega$.

Because of the convex structure of $\ins(\Omega, \hi, \hik)$ we can consider the extreme points of the set.

\begin{definition}
An instrument $\I \in \ins(\Omega, \hi, \hik)$ is \emph{extreme} if a convex sum decomposition $\I = \lambda \J + (1-\lambda) \K$ with some  other instruments $\J, \K \in \ins(\Omega, \hi ,\hik)$ and some $\lambda \in (0,1)$ implies that $\I = \J = \K$.
\end{definition}

The extreme instruments were characterized in \cite{DAPeSe11}.

\begin{proposition}
An instrument $\I \in \ins(\Omega, \hi, \hik)$ with a minimal Kraus decomposition $\I_x(\varrho) = \sum_i K_{ix} \varrho K_{ix}^*$ for all $x \in \Omega$ is extreme if and only if the set $\{K^*_{ix} K_{jx}\}_{i,j,x}$ is linearly independent.
\end{proposition}

\subsection{Simulation scheme}

Similarly to the measurement devices, in the case of quantum instruments, we consider the simulation scheme to be the following: Let $\mathfrak{J}$ be a collection of quantum instruments with outcome set $\Omega$ from $\lh$ to some other output spaces which can be different for different instruments. For any finite subset $\{\J^{(i)}\}_{i=1}^n \subseteq \mathfrak{J}$ we choose an instrument $\J^{(i)} \subseteq \ins(\Omega, \hi, \hik_i)$ with probability $p_i$ according to some probability distribution $(p_i)_{i=1}^n$, measure it, and after obtaining an outcome $(i,x)$ by keeping track of the instrument that we used, we send the output state to another instrument $\mathcal{R}^{(i,x)} \in \ins(\Lambda, \hik_i, \hv)$ according to the classical output $x$ of $\J^{(i)}$. Thus, we obtain a new instrument $\I \in \ins(\Lambda, \hi, \hv)$ defined by
\begin{equation*}
\I_y(\varrho) = \sum_{i=1}^n p_i \sum_{x \in \Omega} \mathcal{R}^{(i,x)}_y\left( \J^{(i)}_x(\varrho) \right)
\end{equation*}
for all $y \in \Lambda$. The set of all instrument obtained from $\mathfrak{J}$ by this method with some finite subset of $\mathfrak{J}$, some probability distribution $(p_i)_i$ and some post-processing instruments $\mathcal{R}^{(i,x)}$ is denoted by $\simu{\mathfrak{J}}$.

In \cite{FiHeLe18} it was shown that in the case of POVMs (and more generally measurements in general probabilistic theories) there is a collection of POVMs that can be used to simulate all other POVMs. Thus, just as in the case of measurements, we can try to reduce the problem of simulability into a specific class of instruments.
\begin{definition}
An instrument $\I$ is simulation irreducible if for any set of instruments $\mathfrak{J}$ such that $\I \in \simu{\mathfrak{I}}$ there exists an instrument $\J \in \mathfrak{J}$ such that $\I \leftrightarrow \J$.
\end{definition}

It is straightforward to verify that the characterization of simulation irreducibility follows the same proof as in \cite{FiHeLe18}.
\begin{proposition}
An instrument is simulation irreducible if and only if it is post-processing clean and post-processing equivalent to an extreme instrument.
\end{proposition}

Since the identity channel is extremal and post-processing clean instruments are exactly those that are equivalent with the identity channel, the set of simulation irreducible instruments coincides with the set of post-processing clean instruments. Thus, we get the following for free from Cor. \ref{cor:identity}.

\begin{proposition}
Every instrument can be simulated by any instrument that is equivalent with the identity channel.
\end{proposition}

Furthermore, it is easy to see that the extreme elements of the equivalence class of the identity channel are just the isometric channels. Namely, if $\I\in \ins(\Omega, \hi, \hik)$ is equivalent with the identity channel, then by Prop. \ref{prop:id-equivalence-class} we have that $\I_x(\varrho) = \sum_{i=1}^{n_x} p_{xi} V_{xi} \varrho V_{xi}^*$ for all $x \in \Omega$ and $\varrho \in \sh$ for some probability distribution $(p_{xi})_{x \in \Omega,i \in \{1, \ldots, n_x\}}$ and some isometries $V_{xi}: \hi \to \hik$ such that $V^*_{xj} V_{xi} = 0$ for all $i \neq j$ for all $x \in \Omega$. We notice that if we define instruments $\mathcal{V}^{(i,y)}\in \ins(\Omega, \hi,\hik)$ by setting $\mathcal{V}^{(i,y)}_{x}(\varrho) = \delta_{xy} V_{yi} \varrho V^*_{yi}$, we see that then $\I = \sum_{y \in \Omega} \sum_{i=1}^{n_y} p_{yi} \mathcal{V}^{(i,y)}$. Thus, $\I$ is extreme if and only if the probability distribution is trivial, i.e., $\I$ has only one outcome so that it is an isometric channel, i.e., $\I(\varrho) = V \varrho V^*$ for all $\varrho \in \sh$ for some isometry $V: \hi \to \hik$.

\section{Summary} \label{sec:summary}

Motivated by the post-processing (and simulability) of POVMs, the first aim of this manuscript is to mathematically correctly define the post-processing of quantum instruments and  characterize the partial order that it induces on the set of equivalence classes of instruments. In particular, we characterize the least and greatest element and characterize their equivalence classes.

We see that similarly to the case of quantum channels in \cite{HeMi13}, every quantum instrument can be post-processed to a so-called trash-and-prepare instrument (the least element) that simply ignores the input state and prepares a new state as output. In accordance with this similarity, we also saw that every instrument can be post-processed from instruments that are equivalent with the identity channel (the greatest element). We find that those instruments consist of randomly chosen isometries with orthogonal output ranges.

Furthermore, we consider two other important classes of instruments. First, the indecomposable instruments have the mathematical advantage that each operation of the instrument consists of only one Kraus operator (for example L\"uders instruments), implying that any instrument can be post-processed from some indecomposable instruments. For the converse we show that any instrument with Kraus operators that have orthogonal output ranges for each outcome (however many Kraus operators it may have) can be post-processed to an indecomposable (detailed) instrument. 

Our conjecture is that this condition is also necessary for an instrument to be equivalent with an indecomposable (detailed) instrument. However, we have not found a general proof for this claim and we leave this as an open question for further work. 

The second class of instruments we study are the measure-and-prepare instruments, which measure some observable on the input  and based on the outcome they prepare a new state. We show how an instrument that can be post-processed from a measure-and-prepare instrument must look like and we show that only measure-and-prepare instruments that have indecomposable (rank-1) induced POVMs are post-processing equivalent with an indecomposable (detailed) instrument. As a corollary we see that actually any instrument with indecomposable induced POVM falls into this class.

For these two classes of instruments, we draw some connections between post-processing of quantum instruments and post-processing of their induced POVMs. In particular, we see that two indecomposable instruments are equivalent if and only if their induced POVMs are equivalent, and that a measure-and-prepare instrument can be post-processed from another measure-and-prepare instrument if and only if the induced POVM of the former instrument can be post-processed from the induced POVM of the latter instrument.

Finally, we use the partial order introduced above to study simulability of instruments. We see that every instrument can be simulated by an instrument from the equivalence class of the identity instrument. We find that the extreme simulation irreducible instruments are the isometric channels.

\acknowledgements{
L.L. acknowledges support from University of Turku Graduate School (UTUGS) and the Academy of Finland via the Centre of Excellence program (Grant No. 312058). M.S. was supported by projects APVV-18-0518 (OPTIQUTE), VEGA 2/0161/19 (HOQIT) and QuantERA project HIPHOP. M.S. was further supported by The  Ministry  of  Education,  Youth  and  Sports  of the   Czech   Republic   from   the   National   Programme of   Sustainability   (NPU   II);   project   IT4 Innovations excellence in science - LQ1602 and through the support of Grant No. 61466 from the John Templeton Foundation, as part of the “The Quantum Information Structure of Spacetime (QISS)” Project (qiss.fr). The opinions expressed in this publication are those of the author(s) and do not necessarily reflect the views of the John Templeton Foundation.}

\end{document}